\documentclass[journal,twoside,web]{IEEEtrans}
\usepackage{amsmath,amssymb,amsfonts} \usepackage{orcidlink}
\usepackage{wrapfig,color}
\usepackage{subcaption}  

\graphicspath{{./},{./figures/}}

\usepackage{dblfloatfix} 
\usepackage{placeins} 

\newcommand{\trace}{{\mathrm{tr}}}

\newtheorem{thm}{Theorem}

\newtheorem{remark}{Remark}
\newtheorem{problem}{Problem}
\def\spacingset#1{\def\baselinestretch{#1}\small\normalsize}
\spacingset{1}

\begin{document}

\title{On the Isospectral Nature of\\Minimum-Shear Covariance Control}

\author{ 
Ralph Sabbagh,   
Asmaa Eldesoukey,
Mahmoud Abdelgalil\orcidlink{https://orcid.org/0000-0003-1932-5115}, and Tryphon T. Georgiou\orcidlink{https://orcid.org/0000-0003-0012-5447}, 
\thanks{Electrical and Computer Engineering, University of California, San Diego, La Jolla, CA, USA, mabdelgalil@ucsd.edu.}
\thanks{Mechanical and Aerospace Engineering, University of California, Irvine, Irvine, CA, USA, tryphon@uci.edu.\newline 
\hspace*{10pt}{Research supported by NSF under ECCS-2347357, AFOSR under FA9550-24-1-0278, and ARO under W911NF-22-1-0292.}
}}

\maketitle

\begin{abstract}
We revisit Brockett's {\em attention} in the context of bilinear gradient flow of an ensemble, and explore an alternative formalism that aims to reduce {\em shear} by minimizing the {\em conditioning} number of the dynamics; equivalently, we minimize the range of the eigenvalues of the dynamics. Remarkably, the evolution is isospectral, and this property is inherited by the coupled nonlinear dynamics of the control problem from a Lax isospectral flow.
\end{abstract}

%

\section{Introduction}

The familiar differential Riccati equation of linear-quadratic regulator theory is the 'archetype' of geometric integrability: a nonlinear differential equation that can be expressed as the 'shadow' of a rotation in higher dimensions. The passage from Riccati dynamics to integrable systems and Lax equations marked a foundational shift in the study of nonlinear dynamics, by revealing hidden linear structure and the conserved quantities that accompany it. In the present work we revisit Brockett's notion of control {\em attention}  \cite{brockett1997minimum}, originally introduced to mitigate sensitivity to state measurements, and to take a fresh look by exploring as an alternative paradigm the minimization of {\em shear} deformation in covariance control.  The bilinear structure of the resulting dynamics leads naturally to a Lax equation for the product of state and costate of the shear-minimization problem. Remarkably, the isospectral property of the Lax equation is inherited by the coupled nonlinear dynamics of the control problem, yielding in particular a conserved spectrum for the strain tensor, which in our setting coincides with the feedback gain matrix.

Our investigation was inspired by Brockett's notion of ``{\em minimum-attention control},'' proposed as a way to quantify implementation costs in engineering practice at an abstract level. As a guiding principle, Brockett suggested that ``the easiest control law to implement is that of a constant input,'' and that ``anything else requires some attention'' \cite{Brockett1997MinAttn,Brockett2003MinimizingAttention}. He proposed to quantify this ``attention'' through the gradient of the control law with respect to spatial and temporal coordinates. That insight has since influenced a broad body of work on resource-aware and event-based control \cite{Jang2015BallCatching,Anta2010MinAttnAnytime,Donkers2014MinAttnLinear,heemels2012introduction,Dirr2016EnsembleMeanVariance,NagaharaNesic2020CDC,nowzari2016distributed,eldesoukey2025collective}, where the timing and frequency of interventions are treated as part of the control design. From one perspective, attention is tied to how often the controller must update; from another, it is tied to how sensitively the control law depends on the information on which it acts. It is this broader interpretation that we adopt here.

From this viewpoint, we propose an alternative to Brockett's quadratic attention penalty by focusing directly on the {\em shear} induced by the control field, or equivalently on the {\em conditioning} of the instantaneous deformation it generates. Concretely, rather than penalizing the Frobenius size of the gain matrix, we seek to reduce the spread of its eigenvalues, i.e., its spectral diameter. This choice is motivated by the observation that a narrow eigenvalue range suppresses anisotropic stretching and compression, and therefore mitigates the sensitivity of the flow to spatial variations. In this sense, minimizing shear provides a complementary formalism for reducing attention: it targets not merely the magnitude of the gain, but the directional disparity in the deformation that the controller imposes on the ensemble.

This perspective also reveals an unexpected integrable structure. The bilinear gradient-flow formulation gives rise to a Lax equation for the product of state and costate variables, and hence to an isospectral evolution. The striking point is that this isospectrality is not confined to an auxiliary lifted system: it is inherited by the nonlinear control dynamics themselves, endowing the evolution with a conserved spectral ``fingerprint.'' Thus, the same mechanism that underlies classical integrable models appears here in a control-theoretic setting, where it constrains the evolution of the strain tensor/feedback gain along optimal trajectories. In this respect, the present work is influenced both by the geometric-control developments of Bloch, Brockett, Marsden, Murray, Ratiu, and others \cite{bloch1992completely,bloch1996euler,bloch2003nonholonomic}, and by our recent work \cite{abdelgalil2025holonomy,abdelgalil2024collective,sabbagh2025minimizing}. The resulting picture is close in spirit to integrable systems such as the Toda lattice \cite{flaschka1974toda,bloch1990convexity}: the flow preserves spectral data while the state itself evolves nontrivially.

\section{Problem formulation}
We consider an ensemble of particles steered under the influence of a time-varying quadratic potential.
The particles are situated at $x_t^{(k)}\in \mathbb R^n$ for $k\in\{1,\ldots,N\}$, and the collection of states
\[
X_t=\begin{bmatrix}x_t^{(1)},&\hdots ,&x_t^{(N)}
\end{bmatrix}
\]
obeys the linear dynamical equation
\[
\dot X_t=A_t X_t, \mbox{ with } A_t\in\mathbb R^{n\times n},
\]
where the time-varying `gain' matrix $A_\cdot$ is our {\em control input}. Individual particles drift in directions $u^{(k)}_t=A_tx^{(k)}_t$, dictated either by the time-varying potential field
\[
U(t,x_t)=\frac{1}{2} \trace (x^\top_t A_t x_t),
\]
or determined individually and locally, via the gain matrix $A$ that is specified centrally and then broadcast to all.

For simplicity we assume that the mean of the ensemble is the origin, and that the control objective is to modify the sample covariance (modulo the normalization by $1/N$)
\[
\Sigma_t = X_tX_t^\top,
\]
from an initial value $\Sigma_0$ to a terminal one  $\Sigma_1$, at $t=1$. Thus, we view $\Sigma_\cdot$ as the state of the ensemble that obeys the ensemble dynamics (differential Lyapunov equation)
\begin{subequations}\label{eq:Lyap_all}
\begin{equation}\label{eq:Lyap}
\dot{\Sigma}_t = A_t \Sigma_t + \Sigma_t A_t, \mbox{ for }t\in[0,1].
\end{equation}
For simplicity of the analysis, we assume throughout that the state $\Sigma_t$ is nonsingular, and that the state space of the ensemble is the cone of positive definite $n\times n$ matrices $\mathbb S_{+}(n)$. We make one further simplification: we assume that $\det(\Sigma_0)=\det(\Sigma_1)$ and that the flow is volume-preserving, i.e., that $\det(\Sigma_t)$ remains constant, or equivalently, that\footnote{When $\det(\Sigma_0)\neq\det(\Sigma_1)$, most reasonable control costs lead to $\trace(A_t)$ being constant and equal to $\frac12 \left(\log(\det\Sigma_1)-\log(\det\Sigma_0)\right)$.}
\begin{equation}\label{eq:traceA}
\trace(A_t)=0, \mbox{ for }t\in[0,1].
\end{equation}
\end{subequations}

When guiding an ensemble of dynamical systems to change their formation, it is desirable to minimize the dependence of the control law on the respective spatial coordinates; this dependence is quantified 
by the spatial derivative
\[
\nabla_x A_tx_t= A_t.
\]
In order to minimize this dependence, Roger Brockett introduced the concept of {\em attention} \cite{brockett1997minimum}, as an integral functional on $A_t$, to reflect how attentive the control needs to be to keep track of the needed actuation. An alternative angle from which we can approach the underlying issue is of space dependence is to view $A_t$ as the {\em shear} tensor of the velocity field that drives the particles, that  quantifies the
relative compression and stretching in different directions. Equivalently, we may consider the incremental state transition map
\[
\exp(A_t \, \delta t)\simeq I+A_t \, \delta t
\]
that dictates the instantaneous deformation of the ensemble, and in this case the {\em conditioning number} is of importance.

To minimize {\em attention}, Brockett advocated the time integral of 
\[
f^{\rm att}(A_t):=\trace (A_t^2)
\]
as a suitable cost functional. On the other hand, focusing on {\em shear} and the {\em conditioning} of the state transition map, it is natural to consider instead the time integral of the spectral range\footnote{Here, $\lambda_{\max},\,\lambda_{\min}$ denote the maximal and minimal eigenvalues.}
\[
f^{\rm cond}(A_t):=\lambda_{\max}(A_t)-\lambda_{\min}(A_t),
\]
that quantifies directly the spread of the eigenvalues of $A_t$. We note that, since
\[
f^{\rm cond}(A) \leq\sqrt{2\,f^{\rm att}(A)}\leq \sqrt{n}\,f^{\rm cond}(A),
\]
these two options, minimizing Brockett's integral attention or the integral spread of the eigenvalues, are expected to mitigate in qualitatively similar ways the sensitivity of the control in the precise knowledge of particle-states.
 
It is of note that, $f^{\rm cond}(\cdot)$ defines a (Minkowski) norm on the space of symmetric and traceless matrices. However, it is of Finsler type, non-differentiable, and thus it is convenient to replace with a smooth surrogate
\[
g_\theta(A):=\frac{1}{\theta}\log\trace(e^{\theta A})
+
\frac{1}{\theta}\log\trace(e^{-\theta A}),
\mbox{ for }\theta>0.
\]
It is standard and easy to show that
\[
f^{\rm cond}(A)\leq g_\theta(A)\leq f^{\rm cond}(A)+\frac{2\log(n)}{\theta},
\]
and thereby, $g_\theta$ approximates the spectral diameter $f^{\rm cond}$ from above, with uniform error $O(\theta^{-1})$.
 
In light of the above, we herein analyze the control-theoretic value of the spectral diameter, as a way to mitigate shear and attention, in steering a collection of Gaussian particles. To this end, we define
\begin{equation}\label{eq:functional}
J_\theta(A_\cdot)
:=
\int_0^1 g_\theta(A_t)^2\,dt,
\end{equation}
and formulate the following problem.
We note that the square in the integrand ensures coercivity of the functional in $L^2$.

\begin{problem}\label{problem}
Given $\Sigma_0,\,\Sigma_1\in \mathbb S_{+}(n)$ with $\det(\Sigma_0)=\det(\Sigma_1)$, establish the existence of a minimizer of
\[
\min\{ J_\theta(A_\cdot) \mid 
(A_\cdot,\Sigma_\cdot) \mbox{ satisfying }\eqref{eq:Lyap_all} \mbox{ with b.c., } \Sigma_0,\,\Sigma_1\}.
\]
\end{problem}

\begin{remark}
The natural classes of matrix functions where to consider solutions of Problem \ref{problem} are\footnote{Here, $\mathbb S$ are the $n\times n$ symmetric matrices, $\mathbb S_0(n):=\{M\in\mathbb S(n)\mid \trace(M)=0\}$, and $\mathbb S_+(n)$ the symmetric positive definite, as before.}
\[
A_\cdot\in L^2([0,1];\mathbb S_0(n)),
\qquad
\Sigma_\cdot\in H^1([0,1];\mathbb S_+(n)).
\]
 On these classes,
$J_{\theta}$ is coercive and weakly lower semicontinuous, and therefore the
problem admits at least one minimizer.
\end{remark}

\section{Analysis}
We begin with the Hamiltonian 
\[
H(A,\Sigma,\Lambda)
=
g_\theta(A)^2+\trace\big(\Lambda(A\Sigma+\Sigma A)\big)
\]
for Problem \ref{problem}, suppressing the time indexing for simplicity.
Setting the variation with respect to $A$ (where $A\in\mathbb S_0(n)$) equal to zero gives  
\begin{equation}\label{eq:M}
g_\theta(A)G_\theta(A)+ M=0,
\end{equation}
for the `momentum' matrix
\begin{equation}\label{eq:MM}
M
:=
\frac12(\Sigma\Lambda+\Lambda\Sigma)
-\frac{1}{2n}\trace(\Sigma\Lambda+\Lambda\Sigma)I,
\end{equation}
and
$$
G_\theta(A):=\nabla_A g_\theta(A)
=
\frac{e^{\theta A}}{\trace(e^{\theta A})}
-
\frac{e^{-\theta A}}{\trace(e^{-\theta A})}.
$$
The variation with respect to $\Sigma$ gives the costate equation
\[
\dot\Lambda=-(\Lambda A+A\Lambda).
\]
The product of state with costate
\[
L_t:=\Lambda_t\Sigma_t
\]
satisfies
\begin{align}\nonumber
\dot L
&=
\dot\Lambda\,\Sigma+\Lambda\,\dot\Sigma\\\nonumber
&=
(-\Lambda A-A\Lambda)\Sigma+\Lambda(A\Sigma+\Sigma A)\\\nonumber
&=
\Lambda\Sigma A-A\Lambda\Sigma,
\end{align}
which is in the isospectral Lax form
\begin{equation}\label{eq:Lax}
\dot L=[L,A]
\end{equation}
with $[L,A]=LA-AL$ denoting the commutator. Thus, we have apparently arrived at an integrable system (see below).

Let us recap. With $M$ being the (traceless) symmetric part of $L$, and $\Omega$ the anti-symmetric, re-introducing the time-indexing,
\begin{align}\label{Ldecomp}
L_t=M_t+\Omega_t+\frac{1}{n}\trace(L_t)I,
\end{align}
and the equations of motion become
\begin{subequations}\label{eq:main}
\begin{align}
\dot\Sigma_t&=A_t\Sigma_t+\Sigma_tA_t\label{eq:a}\\
\dot M_t&=\Omega_tA_t-A_t\Omega_t\label{eq:b}\\
\dot\Omega_t&=M_tA_t-A_tM_t\label{eq:c}
\end{align}
and since
\begin{align} \label{eq:A}
M_t=-g_\theta(A_t)G_\theta(A_t)
\end{align}
is a function of $A_t$, $M_tA_t-A_tM_t=0$ and \eqref{eq:c} becomes
\begin{equation}
\dot \Omega = 0 \tag{\ref{eq:c}'}. \label{eq:cp}
\end{equation}
Thus, $\Omega_t=\Omega$ remains constant throughout.
\end{subequations}

\begin{remark}\label{rem2}
The above system of equations (\ref{eq:a},\ref{eq:b},\ref{eq:cp},\ref{eq:A}) specify a two-point boundary value problem, with boundary conditions $\Sigma_0$ and $\Sigma_1$. It can be numerically solved using a shooting method, starting from $(\Sigma_0,L_0)$ and computing $A_t$ from the running value of $M_t$ using \eqref{eq:A}. The existence of a value $L_0$, so that the terminal value $\Sigma_1$ matches the given boundary condition specified is explained next.
\end{remark}

\begin{remark}\label{rem3}
    The computation of $A_t$ from $M_t$ using \eqref{eq:A} amounts to solving $n$ scalar transcendental equations
   \begin{equation}\label{eq:transcendental}
   \mu_i = -g_\theta(\lambda) \left( \frac{e^{\theta \lambda_i}}{\sum_j e^{\theta \lambda_j}} - \frac{e^{-\theta \lambda_i}}{\sum_j e^{-\theta \lambda_j}} \right),
   \end{equation}
    subject to the constraint $\sum \lambda_i = 0$,
    for the eigenvalues $\lambda_i$ of $A_t$, from those $\mu_i$ of $M_t$, as $A_t$ and $M_t$ share the same eigenvectors.
\end{remark}

\begin{remark}\label{constancyofeigenvalues}
    It is a rather remarkable fact that the system of equation (\ref{eq:a},\ref{eq:b},\ref{eq:cp},\ref{eq:A}) inherits from the isospectral flow \eqref{eq:Lax} the constancy of eigenvalues. Specifically,
    from \eqref{eq:Lax} we readily see that the eigenvalues of $L_t$ remain constant in time. Since the skew symmetric part of $L_t$, $\Omega$, is constant, the spectrum of the (traceless) symmetric part $M_t$ is also constant. Hence, the eigenvalues of $A_t$ also remain constant, and the system of equations \eqref{eq:transcendental} only needs to be solved at the start of the interval in the shooting method of Remark \ref{rem2}.
\end{remark}

\begin{thm}\label{theorem1}
Given positive definite matrices $\Sigma_0$, $\Sigma_1$ with $\det(\Sigma_0)=\det(\Sigma_1)$, there exists a value $L_0$ so that the minimizer of Problem \ref{problem} can be obtained by integrating the system \eqref{eq:a}--\eqref{eq:cp} together with \eqref{eq:A}.
\end{thm}

The cost functional \eqref{eq:functional} is convex with respect to $A_\cdot$. Moreover, the dynamics $\dot\Sigma=A\Sigma+\Sigma A$ are controllable on the manifold of positive definite matrices with fixed determinant, ensuring that at least one feasible path exists connecting any valid $\Sigma_0$ and $\Sigma_1$. Indeed, the `super-operator'
\begin{equation}\label{super}
\dot\Sigma_t \mapsto A_t=\int_0^\infty e^{-\Sigma_t \tau}\dot\Sigma_t e^{-\Sigma_t \tau}d\tau,
\end{equation}
solves $\dot\Sigma=A\Sigma+\Sigma A$ for $A_t$ given $\dot\Sigma_t$ and $\Sigma_t$, and is onto on the tangent space of the fixed-determinant leaf. However, the mapping $A_\cdot\mapsto\Sigma_1$ is not linear. Because of that, the uniqueness of the minimizer cannot be guaranteed, in general. The detailed proof follows.
\begin{figure*}[!ht]
    \centering
    \includegraphics[width=\linewidth]{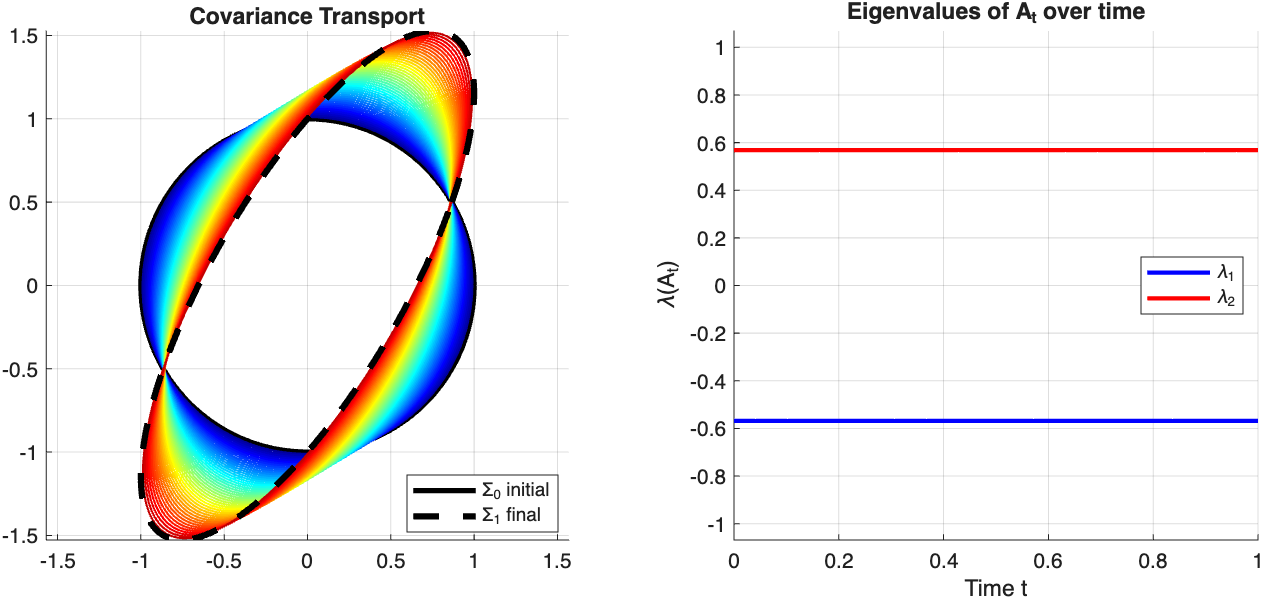}
    \caption{Simulation of the planar covariance transport solving the two-point boundary value problem via the shooting method. \textbf{Left:} The optimal trajectory of the covariance ensemble flowing from the initial state $\Sigma_0$ (solid black) to the target state $\Sigma_1$ (dashed black). The intermediate states are colored chronologically, illustrating the smooth, volume-preserving deformation under the spectral spread cost. \textbf{Right:} The eigenvalues of the optimal stretch matrix $A_t$ over time $t \in [0,1]$. Consistent with the integrable Lax structure and the isospectral flow dynamics established in the analysis, the eigenvalues remain strictly constant throughout the transport.}
    \label{fig:covariance_transport_comparison}
\end{figure*}
{\em Proof of Theorem \ref{theorem1}}:
First, there exists an admissible control with finite cost. Indeed, if
\[
\Phi=\Sigma_1^{1/2}\Big(\Sigma_1^{1/2}\Sigma_0\Sigma_1^{1/2}\Big)^{-1/2}\Sigma_1^{1/2},
\]
then taking $A_t=\log(\Phi)$ constant on $[0,1]$ gives $J_{\theta}(A_\cdot)<\infty$. Moreover,
\[
\det(\Phi)^2=\frac{\det(\Sigma_1)}{\det(\Sigma_0)}=1,
\]
so
\[
\trace(A_t)=\log\det(\Phi)=0,
\]
and the control is traceless.

Next, we verify coercivity in $L^2$. Let
\[
s(A):=\lambda_{\max}(A)-\lambda_{\min}(A).
\]
For symmetric $A$, we have $g_\theta(A)\ge s(A)$. Since $A$ is traceless, all eigenvalues lie in the interval $[\lambda_{\min}(A),\lambda_{\max}(A)]$, and therefore
\[
\trace(A^2)\le n\,s(A)^2\le n\,g_\theta(A)^2.
\]
Hence
\begin{equation}\label{eq:spectral-l2}
J_\theta(A_\cdot)\ge \frac1n\int_0^1 \trace(A_t^2)\,dt.
\end{equation}
Thus the sublevel sets of $J_\theta$ are bounded in $L^2$, and any minimizing sequence $\{A_\cdot^{(k)}\}$ admits a subsequence, not relabeled, such that
\begin{equation}\label{eq:weakconv}
A_\cdot^{(k)}\rightharpoonup A_\cdot
\quad \text{weakly in }L^2([0,1];\mathbb{S}_0(n)).
\end{equation}

By Gr\"onwall,
\begin{align*}
\|\Phi_t^{(k)}\|
&\le
\exp\Big(\int_0^t \|A_s^{(k)}\|\,ds\Big)
\le
\exp(\|A^{(k)}\|_{L^1})
\\&\le
\exp(\|A^{(k)}\|_{L^2}),
\end{align*}
so $\{\Phi_\cdot^{(k)}\}$ is uniformly bounded on $[0,1]$. Also,
\[
\|\dot\Phi_\cdot^{(k)}\|_{L^1}
\le
\|A_\cdot^{(k)}\|_{L^1}\|\Phi_\cdot^{(k)}\|_{L^\infty}
\le
\|A_\cdot^{(k)}\|_{L^2}\|\Phi_\cdot^{(k)}\|_{L^\infty},
\]
hence $\{\Phi_\cdot^{(k)}\}$ is bounded in $W^{1,1}([0,1],\mathbb R^{n\times n})$. By Arzel\`a--Ascoli, a further subsequence converges uniformly to a continuous limit $\Phi_\cdot$. Standard arguments then show that $\dot\Phi_t=A_t\Phi_t$, $\Phi_0=I$, and $\Phi_t\Sigma_0\Phi_t^\top$ satisfies the end-point condition at $t=1$. Since $A_\cdot\in L^2$ and $\Phi_\cdot$ is bounded, the corresponding state $\Sigma_\cdot=\Phi_\cdot\Sigma_0\Phi_\cdot^\top$ belongs to $H^1([0,1],\mathbb S(n))$.

Finally, the map
\[
A\mapsto \frac{1}{\theta}\log\trace(e^{\theta A})
\]
is convex on $\mathbb S(n)$, and so is
\[
A\mapsto \frac{1}{\theta}\log\trace(e^{-\theta A}).
\]
Hence $A\mapsto g_\theta(A)$ is convex and nonnegative, and therefore $A\mapsto g_\theta(A)^2$ is convex as well. It follows that
\[
\int_0^1 g_\theta(A_t)^2\,dt
\le
\liminf_{k\to\infty}
\int_0^1 g_\theta(A_t^{(k)})^2\,dt.
\]
Thus $J_\theta$ is weakly lower semicontinuous. The weak limit $A_\cdot$ is feasible and achieves the minimum. \hfill $\Box$

\begin{remark}
Uniqueness of the minimizer is not guaranteed. While the cost $J_\theta$ is convex, the admissible set is not, due to the nonlinear dependence of the endpoint constraint $\Sigma_1$ on the control $A_\cdot$. Furthermore, the problem inherits the orthogonal symmetry of the boundary data: any orthogonal matrix $Q$ preserving both $\Sigma_0$ and $\Sigma_1$ generates an equivalent minimizer $(Q^\top A_t Q, Q^\top \Sigma_t Q)$ with identical cost, since $g_\theta$ is invariant under orthogonal conjugation. Thus, uniqueness can at best be expected modulo the common symmetry group of the boundary covariances.
\end{remark}
\section{Example}

We implemented the shooting method to solve the two-point boundary value problem \eqref{eq:main} as explained in Remarks \ref{rem2} and \ref{rem3}.
In this, we find an initial value for $L_0$ that steers the covariance from $\Sigma_0$ to $\Sigma_1$ following the integrable Lax dynamics $\dot{L} = [L, A]$. A representative path is shown in Fig. \ref{fig:covariance_transport_comparison}.

\bibliographystyle{plain}
\bibliography{References}
\end{document}